\RequirePackage{fix-cm}
\documentclass[twocolumn]{svjour3}          
\smartqed  
\usepackage{graphicx}
\usepackage{xspace}
\usepackage{amssymb,amsmath}
\usepackage{url}
\usepackage[]{hyperref}

\newtheorem{observation}{Observation}
\newtheorem{thmone}{Theorem}

\def\push{\texttt{push}\xspace}
\def\pull{\texttt{pull}\xspace}
\def\ex{\texttt{push\&pull}\xspace}
\providecommand{\abs}[1]{\lvert#1\rvert}
\def\whp{w.h.p.\xspace}

\begin{document}

\title{Breaking the $\log n$ Barrier on Rumor Spreading
\thanks{
An extended abstract of this work appeared in \cite{avin2013faster}. The work of the 
second author was partially 
supported by the Austrian Science Fund (FWF) under contract P25214-N23 \emph{``Analysis of 
Epidemic Processes and Algorithms in Large Networks''}. The main result of this paper
solves an open problem 
presented at Dagstuhl Seminar 13042 \emph{``Epidemic Algorithms and Processes: From Theory to 
Applications''}.}
}


\author{Chen Avin         \and
            Robert Els\"asser 
}


\institute{Chen Avin \at
             Communication Systems Engineering\\
              Ben Gurion University of the Negev, Israel\\
              \email{avin@cse.bgu.ac.il}           
           \and
           Robert Els\"asser \at
            Department of Computer Sciences\\
            University of Salzburg, Austria \\
            \email{elsa@cosy.sbg.ac.at} 
 }

\date{Received: date / Accepted: date}


\maketitle

\begin{abstract}
$O(\log n)$ rounds has been a well known upper bound for rumor spreading using \ex in 
the \emph{random phone call} model (i.e., uniform gossip
in the complete graph). A matching lower bound of $\Omega (\log n)$ is also known for 
this special case.
Under the  assumption of this model and with a natural addition that nodes can call a 
partner once they learn its address (e.g., its IP address) we present a new distributed, 
address-oblivious and robust algorithm that uses \ex with pointer jumping to spread a 
rumor to all nodes in only $O(\sqrt{\log n})$ rounds, w.h.p. This algorithm can 
also cope with $F= 
O(n/2^{\sqrt{\log n}})$ node failures, in which case all but $O(F)$ nodes become 
informed within $O(\sqrt{\log n})$ rounds, w.h.p. 
\end{abstract}



\section{Introduction}
Gossiping, or rumor-spreading, is a simple stochastic process for dissemination of information across a network. 
In a \emph{round} of gossip, \emph{each} node chooses a single, usually random, neighbor as its \emph{communication partner} according to a \emph{gossip algorithm} (e.g., selecting a random neighbor). Once a partner is chosen the node \emph{calls} its partner and a limited amount of data is transferred between the partners, as defined by the gossip \emph{protocol}. Three basic actions are considered in the literature: either the caller pushes information to its partner (\push), pulls information from the partner (\pull), or does both (\ex). In the most basic information dissemination task, a token or a rumor in placed arbitrary in the network and we are interested in the number of rounds and message transmissions until all nodes in the networks receive the rumor.
The selection of the protocol can lead to significant differences in the performance. Take for example the star graph, let nodes call a neighbor selected uniformly at random and assume the rumor is placed at one of the leafs. It is easy to see that both \push and \pull will require $\omega(n)$ rounds to complete the spreading of a single rumor while \ex will take only two rounds.

Somewhat surpassingly, but by now well understood, randomized rumor-spreading turned out to be very efficient in terms of time and message complexity while keeping robustness to failures \cite{karp2000randomized,ES06}. In addition, this type of algorithms are very simple and distributed in nature so it is clear why gossip protocols have gained popularity in recent years and have found many applications both in communication networks and social networks.
To name a few examples: updating a database replicated at many sites~\cite{demers1987epidemic,karp2000randomized}, resource discovery~\cite{harchol1999resource}, computation of aggregate information~\cite{kempe03gossip}, multicast via network
coding~\cite{deb06algebraic}, membership services~\cite{gurevich2010correctness}, or the spread of influence and gossip in social networks~\cite{Kempe2003Maximizing,Chaintreau2008Opportunistic}.

In this paper we consider the most basic scenario, the \emph{random phone call model}~\cite{karp2000randomized}, where the underlying network is the complete graph and nodes can call a random neighbor according to some given distribution. In addition, the model requires the algorithm to be \emph{distributed} and \emph{address-oblivious}: it cannot use the address of the current communication partners to determine its state (for an exact definition see Section \ref{prelim}). For example this setting fits well to applications which require communication over the internet such as peer-to-peer protocols and database synchronization. A node can pick and call any (random or given) neighbor via its IP address, but it is desired to keep the algorithm address-oblivious otherwise it may have critical points of failure. For example agreeing before hand on a leader to contact (by its IP address) is \emph{not} an address-oblivious algorithm. Furthermore, such a protocol 
is also highly fragile, although it leads to efficient information spreading (as pointed out 
in the star graph example above).

The random phone call model was thoroughly studied in the literature starting with the work of Frieze and Gimmet~\cite{frieze1985shortest} and following by Pittel~\cite{pittel1987spreading} who proved an upper bound of $O(\log n)$ rounds  for \push in the complete graph. Demers et al.~\cite{demers1987epidemic} considered both \push and \pull  as a simple and decentralized way to disseminate information in a network and studied their rate of progress. Finally, Karp et al.~\cite{karp2000randomized} gave a detailed analysis for this model. They used \ex to optimize the message complexity and showed the robustness of the scheme. They proved that while using only \push the communication overhead is  $\Omega(n \log n)$, their algorithm 
only requires $O(n \log \log n)$ message transmissions by having a  
running time of $O(\log n)$, even under arbitrary oblivious failures.  Moreover they proved that any address-oblivious algorithm (that selects neighbors uniformly at random) will require  $\Omega(n \log \log n)$ message transmissions. 

\subsection{Our contribution}

We consider the same assumptions as in the random phone call model: the algorithm needs to be distributed, address-oblivious and it can select neighbors at random. In addition we use the fact that given an address of a node (e.g., its IP address) the caller can call directly on that address. 
This slight addition leads to a significant improvement in the number of rounds from $O(\log n)$ to  $O(\sqrt{\log n})$, but still keeps the algorithm robust. Furthermore, assume that a node may fail (at the beginning or during the algorithm is executed) with probability 
$O(1/2^{\sqrt{\log n}})$, independently. The main result of the paper is the following theorem:

\begin{thmone}
\label{thm:intro}
At the end of the algorithm Jumping-Push-Pull (JPP), all but $O(F)$ nodes are informed \whp\footnote{In this paper with high probably or \whp is with probability at least $1-n^{-1-\Omega(1)}$}, where $F$ is the number of failed nodes (as described above). 
The algorithm has running time $O(\sqrt{\log n})$
and produces a {\bf bit} communication complexity of 
$O(n (\log^{3/2} n + b \cdot \log \log n))$, \whp, where $b$ is the 
bit length of the message.  
\end{thmone}

Clearly, if there are no failures (i.e., $F=0$), then all nodes become informed in
the number of rounds given in Theorem \ref{thm:intro}.
As mentioned, we inform all nodes in $O(\sqrt{\log n})$ rounds vs.~$O(\log n)$ rounds achieved by 
the algorithm of Karp et al. Our message complexity is $O(n\sqrt{\log n})$ compared to  $O(n \log \log n)$ and if the rumor is of bit length $b = \Omega(\frac{\log^{3/2} n}{\log \log n})$ both of the algorithms bit complexity is $\Omega(b \cdot n \log \log n)$. Moreover, if there are $\Omega(n)$ 
messages to be distributed in the network, then the first term in the 
expression describing the bit communication complexity is amortized over the total number 
of message transmissions (cf.~\cite{karp2000randomized}), and we obtain the 
same communication overhead as in \cite{karp2000randomized}.

Few words on the basic idea of the algorithm are in place. In a nutshell our approach has two phases: first we try to build an infrastructure, a virtual topology, that is efficient for \ex. Second, we perform a simple \ex on the virtual topology. The running time is the combination of both these tasks.
For example, constructing a random star would be preferable since the second phase will then take only a constant number of rounds, but as it turns out the cost of the first phase, in this case, is too high. Interestingly, our algorithm results in balancing these two phases where each task requires $O(\sqrt{\log n})$ rounds. Instead of a star with a single leader we build a virtual topology with about random $n/2^{\sqrt{\log n}}$ leaders and each leader is connected to about $2^{\sqrt{\log n}}$ nodes we call \emph{connectors} (a node is either a leader or a connector). Each connector is then linked to two leaders after a process of pointer jumping~\cite{leighton1992introduction} . This simple 2-level hierarchy results in a very efficient information spreading. Leaders are a source of fast \pull mechanism and connectors are essential for fast spreading among leaders using \push. Our approach was motivated from similar phenomena in social networks \cite{fountoulakis2012ultra,Avin2012From}
(see the related work section for a more detailed description of these results).

\noindent {\bf Journal version update:} Motivated by the conference version of this paper~\cite{avin2013faster}, Haeupler and Malkhi~\cite{haeupler14optimal} improved our bound and presented an elegant algorithm that solves the problem we study here in  $O(\log \log n)$ rounds together with a macthing lower bound. Nevertheless we think our work contributes to the understanding of the gossiping process and may be useful in extension of the model to general graphs.




\section{Preliminaries - Rumor Spreading}
\label{prelim}
Let $G(V, E)$ be an undirected graph, with $V$ the set of nodes and
$E$ the set of edges. Let $n = \abs{V}$ and $m = \abs{E}$. For $v \in V$, let $N(v) = \{u \in V \mid (vu) \in E\} $ the
set of neighbors of $v$ and $d(v) = \abs{N(v)}$ the degree of
$v$.
Initially a single arbitrary node holds a rumor (i.e., a token) of size $b$ bits; then the process of rumor-spreading (or gossiping) progresses in synchronous \emph{rounds}. At each round, each node $v$ selects a single \emph{communication partner}, $u \in N(v)$ from its neighbors and $v$ calls $u$. The method by which $v$ choses $u$ is called  the \emph{goosip algorithm}. 
The algorithm is called \emph{address-oblivious} if $v$'s state in round $t$ does not depend on the addresses of its communication partners at time $t$. 
Meaning, any decision about if, how and what to send in the current round is made before the current round. Nevertheless, $v$'s state can still depend on the addresses of its communication partners from previous rounds~\cite{karp2000randomized}. 

Randomized gossip is maybe the most basic address-oblivious algorithm, in particular, when the communication partners are selected uniformly at random the process is known as \emph{uniform gossip}.
A well studied such case is the \emph{random phone call model}~\cite{karp2000randomized} where $G$ is the complete graph and $u$ is selected u.a.r from $V \setminus v$. 
Upon selecting a communication partner the \emph{gossip protocol} defines the way and which information is transferred between $v$ and $u$. Three basic options are considered to deliver information between communication partners: \push, \pull and \ex. In \push the calling node, $v$, sends  a message 
to the called node $u$, in \pull a message is only transferred the other way (if the called node, $u$, has what to send) and in \ex each of the communication partners sends a message to the node at the other end of the edge. The content of the messages is defined by the protocol and can contain only the rumor (in the simplest case) or additional information like counters or state information (e.g., like in~\cite{karp2000randomized}). 

After selecting the graph (or graph model), the gossip algorithm and protocol, the main metrics of interest are the dissemination time and the message complexity. Namely how many rounds and messages are needed until all vertices are informed\footnote{a call, in which no data is sent (e.g., the rumor, or a pointer), is not considered as a message}  (on average or with high probability), even under node failures.
The \emph{bit complexity} is also a metric of interest and counts the \emph{total} number of bits sent during the dissemination time. This quantity is a bit more involved since it depends also on $b$ (the size of the rumor) and messages at different phases  of the algorithms may have different sizes. 


A \emph{pointer jumping} is a classical operation from parallel algorithm design~\cite{leighton1992introduction} where the destination of your next round pointer is the pointer at which your current pointer points to. Our algorithm uses pointer jumping by sending the addresses (i.e., pointers) of previous communication partners to current partners (see Section \ref{analysis} for a detailed description).


\section{Related Work}
Beside the basic random phone call model, gossip algorithms and rumor spreading were generalized in several different ways.
The basic extension was to study \emph{uniform gossip} (i.e., the called partner is selected uniformly at random from the neighbors lists) on graphs other than the clique.
Feige et. al.~\cite{feige90randomized} studied randomized broadcast in networks and extended the result of $O(\log n)$ rounds for \push to different types of graphs like hypercubes and random graphs models. 
Following the work of Karp et al.~\cite{karp2000randomized}, and in particular in recent years the \ex protocol was studied intensively, both to give tight bounds for general graphs and to understand its performance advantages on specific families of graphs. 
A lower bound of $\Omega(\log n)$ for uniform gossip on the clique can be conclude from \cite{sauerwald2010mixing} that studies the sequential case. We are not aware of a lower bound for general, address-oblivious \ex.

Recently Giakkoupis~\cite{giakkoupis2011tight} proved an upper bound for general graphs as a function of the \emph{conductance}, $\phi$, of the graph, which is $O(\phi^{-1} \log n)$ rounds. Since the conductance is at most a constant this bound cannot lead to a value of $o(\log n)$, but is tight for many graphs. Doerr et al.~\cite{doerr2011social} studied information spreading on a known model of social networks and showed for the first time
an upper bound which is $o(\log n)$ for a family of natural graphs. They proved that while uniform gossip with \ex results in $\Theta(\log n)$ rounds in preferential attachment graphs, a slightly improved version where nodes are not allowed to repeat their last call results in a spreading time of $O(\frac{\log n}{\log\log n})$. A similar idea was previously used in \cite{ES08,BET08} to reduce the message complexity of \ex in random graphs.
Fountoulakis et al.~\cite{fountoulakis2012ultra} considered spreading  arumor to all but a small $\epsilon$-fraction of the population. For random power law graphs \cite{CL02} they proved that \ex informs all but an $\epsilon$-fraction of the nodes in  $O(\log \log n)$ rounds. Their proof relies on the existence of many 
\emph{connectors} (i.e., nodes with low degree connected to high degree nodes) 
which amplify the spread of the rumor between high degree nodes, and this influenced our approach; in some sense our algorithm tries to imitate the structure of the social network they studied.


Another line of research was to study \ex (as well as \push and \pull separately) but not under the uniform gossip model.
Censor-Hillel et al.~\cite{censor2012global}, gave an algorithm for all-to-all dissemination in arbitrary graphs which eliminates the dependency on the conductance.
For unlimited message sizes (essentially you can send everything you know), their randomized algorithm informs all nodes in $O(D + \mathrm{polylog}(n))$ rounds where $D$ is the graph diameter;
clearly this is tight for many graphs. 
Quasirandom rumor spreading was first offered by Doerr et al. in~\cite{doerr2008quasirandom,DFS09} and showed to outperform the randomize algorithms in some cases (see also \cite{BES15} for a study of the message complexity of quasirandom rumor spreading). Most recently Haeupler~\cite{Haeupler2013Simple} proposed a completely deterministic algorithm that spread a rumor with $2(D+\log n)\log n$ rounds (but also requires unlimited message size). 

In a somewhat different model (but similar to ours), where nodes can contact any address as soon as they learn about it, Harchol-Balter et. al.~\cite{harchol1999resource} considered the problem of resource discovery (i.e., learning about all nodes in the graph) starting from an arbitrary graph. 
They used a form of one hop pointer jumping with \ex and gave an upper bound of $O(\log^2 n)$rounds for their algorithm.
Kutten at. el. \cite{kutten2003deterministic,kutten2007asynchronous} studied resource discovery both in the deterministic and the asynchronous cases and presented improve bound.

The idea of first building a virtual structure (i.e.; topology control) and then do gossip on top of this structure is not novel and similar idea was presented by Melamed and Keidar  \cite{melamed2004araneola}.
Another source of influence to our work was the work on pointer jumping with \ex in the context of efficient construction of peer-to-peer networks~\cite{mahlmann2006distributed}
and on computing minimum spanning tress \cite{lotker2005minimum}.

\section{Jumping-Push-Pull in $O(\sqrt{\log n})$-time}
\label{analysis}
First, we present the algorithm, which disseminates a rumor by 
\ex in $O(\sqrt{\log n})$ time, w.h.p. Then, 
we analyze our algorithm, show its corectness, and 
prove the runtime bound.

\subsection{Algorithm - Rumor Spreading with Pointer Jumping}
\label{algo}
First, we provide a high-level overview of our algorithm.
At the beginning, a message resides on one of the nodes, and 
the goal is to distribute this message (or rumor) to every node 
in the network.
We assume that each node has a unique address (which can e.g.~be
its IP-address), and every node can select a vertex uniformly
at random from the set of all nodes (i.e., like in the random phone call model). 
Additionally, a node  can store a constant number of addresses, out of which 
it can call one of them in a future round. However, a node must decide 
in each round whether it chooses an address uniformly at random 
or from the pool of the addresses stored before the current 
round.  

In our analysis, we assume for simplicity 
that every node knows $n$ exactly. However, a slightly modified version of our algorithm also 
works if the nodes have an estimate of $\log n$, which is correct up to some 
constant factor. We discuss this case in Section \ref{disc}.

The algorithm consists of five main \emph{phases} and these phases may
contain several \emph{rounds} of communication.
Basically there are two type of nodes in the algorithm, which we call  
\emph{leaders} and \emph{connectors}, and the algorithm is:

\begin{itemize}
\item[] \textbf{Phase 0} - each informed node performs \push in every step of this phase. 
The phase consists of $c \log \log n$ steps, where $c$ is some suitable constant.
According to e.g.~\cite{karp2000randomized}, the message is contained in 
$\log^2 n$ many nodes at the end of this phase.
\item[] \textbf{Phase 1} - each node flips a coin to decide whether 
it will be a leader, with probability $1/2^{\sqrt{\log n}}$, 
or a connector, with probability $1-1/2^{\sqrt{\log n}}$.
\item[] \textbf{Phase 2} - each connector chooses leaders by preforming five 
pointer jumping sub-phases, each for $c \sqrt{\log n}$ rounds. At the end, 
all but $o(n)$ connectors will have at least 2 leader addresses stored 
with high probability.
Every such connector keeps exactly $2$ leader addresses 
(chosen uniformly at random) and forgets all the others. A detailed description of this 
phase is given below.
\item[] \textbf{Phase 3} - each connector opens in each round of this phase 
a communication channel to a randomly chosen 
node from the list of leaders received 
in the previous phase. However, once a connector receives the 
message, it only transmits once in the next round using \push communication to its other leader. 
The leaders send the message 
in each round over all incoming channels during the whole phase (i.e., the leaders 
send the message by 
\pull). The length of this phase is $c \sqrt{\log n}$ rounds.
\item[] \textbf{Phase 4} - every node performs the usual 
\ex (median counter algorithm according to \cite{karp2000randomized}) for 
$c \sqrt{\log n}$ rounds. All informed nodes are considered to be 
in state $B_1$ at the beginning of this phase (cf.~\cite{karp2000randomized}). 
\end{itemize}


The second phase needs some clarification: 
it consists of $5$ sub-phases in which connectors chose leaders. 
In each sub-phase,  every connector performs so called
pointer-jumping~\cite{leighton1992introduction} for $c \sqrt{\log n}$ rounds, where $c$ 
is some large constant. 
The leaders do not participate in pointer jumping, and 
when contacted by a connector, they let it know that it has reached a leader.
The pointer jumping sub-phase works as follow: in the first round 
every connector chooses a node uniformly at random,
and opens a communication channel to it. Then, each (connector or leader) node, 
which has incoming communication channels, sends its 
address by \pull to the nodes at the other end of these channels.
In each round $i>1$ of this sub-phase, every connector 
calls on the address obtained in step $i-1$, and opens a 
channel to it. Every node, which is incident to 
an incoming channel, transmits the address obtained in 
step $i-1$. Clearly, at some time $t$ each node stores only the address received in the 
previous step $t-1$ of the current sub-phase, and the addresses stored 
at the end of the previous sub-phases.
If in some sub-phase a connector $v$ does not receive a leader address at all, 
then it forgets the address stored in the last step of
this sub-phase. In this case we say that $v$ is ``black'' in this sub-phase.
The idea of using connectors to amplify the information propagation 
in graphs has already been used in e.g.~\cite{fountoulakis2012ultra}. 



From the description of the algorithm it follows that its 
running time is $O(\sqrt{\log n})$. In the next section we 
show that every node becomes informed with probability 
$1-n^{-1-\Omega(1)}$.

 
\subsection{Analysis of the Algorithm}
For our analysis we assume the following failure model. Each node may fail
(before or during the execution of the algorithm) with some probability 
$O(1/2^{\sqrt{\log n}})$. This implies that  
e.g.~$n^{1-\epsilon}$ nodes may fail in total, where $\epsilon >0$ can be any small 
constant. If a node fails, then it does not participate 
in any pointer- or message-forwarding process. Moreover, we assume that the 
other nodes do not realize that a node has failed, even if they contact him 
directly. That is, all nodes which contact (directly or by 
pointer-jumping) a failed node in some sub-phase are 
also considered to be failed.

First, we give a high-level overview of our proofs. Basically, we do not consider phases $0$ 
and $1$ in the analysis; the resulting properties on the set of informed nodes are 
straight-forward, and have already been discussed in e.g.~\cite{karp2000randomized}. 
Thus, we know that at the end of phase $0$, the rumor is contained in at least 
$\log^2 n$ nodes, and at the end of phase $1$ there are $n/2^{\sqrt{\log n}} \cdot 
(1 \pm o(1))$ leaders, \whp Lemma \ref{cycle} analyzes phase $2$. We show that most of the 
connectors will point to a leader after a sub-phase, \whp To show this, we bound the probability
that for a node $v$, the choices of the nodes in the first step of this sub-phase lead 
to a cycle of connectors, such that after performing pointer jumping for $c \sqrt{\log n}$ 
steps, $v$ will point to a node in this cycle. Since we have in total $5$ sub-phases, which are 
run independetly, we conclude that each connector 
will point to a leader, after at least $2$ sub-phases. At this point we do not consider node failures.

In Lemma \ref{upper-bound}, we basically bound the number of nodes pointing to the same 
leader. For this, we consider the layers of nodes, which are at distance $1$, $2$, etc... from 
an arbitrary but fixed leader $u$ after the first step of a sub-phase. Since we know how many 
layers we have in total, and bound the growth of a layer $i$ compared to the previous layer $i-1$
by standard balls into bins techniques, we obtain an upper bound, which is polynomial 
in $2^{\sqrt{\log n}}$. 

In Lemma \ref{no-followers} we show that most of the connectors share a leader 
address at the end of a sub-phase with 
$\Omega(2^{\sqrt{\log n}}/\log n)$ many connectors, \whp Here, we start to consider node failures too.
To show this, we compute the expected 
length of the path from a connector to a leader after the first step of a sub-phase. However, 
since these distances are not independent, we apply Martingale techniques to show that 
for most nodes these distances occur with high probability.

Lemma \ref{lem-growth} analyzes then the growth in the number of informed nodes
within two steps of phase $3$. What we basically show is that after any two steps, 
the number of informed nodes is increased by a factor of $2^{\sqrt{\log n}/2}$, \whp, 
and most of the newly informed nodes are connected to a (second) leader, which is not 
informed yet. Thus, most connectors which point to these leaders are also not informed.
These will become informed two steps later.

The main theorem then uses the fact that at the end of phase $3$ a $2^{7 \sqrt{\log n}}$
fraction of the nodes is informed, \whp Then, we can apply the algorithm of 
\cite{karp2000randomized} to inform all nodes within additional $O(\sqrt{\log n})$ steps, \whp

Now we start with the details.
In the first lemma we do not consider node failures.
For this case, we show that, \whp,  there is no connector which 
is ``black'' in more than two sub-phases of the second phase.
Let $r(v)$ be the choice of an arbitrary but fixed connector node $v$
in the first round of a sub-phase. Furthermore, let 
$R(v)$ be the set of nodes which can be reached 
by node $v$ using (directed) edges of the form $(u,r(u))$
only. That is, a node $u$ is in $R(v)$ iff 
there exist some nodes $u_1, \dots , u_k$ 
such that $u_1 = r(v)$, $u_{i+1} = r(u_i)$ for any 
$i \in \{ 1, \dots , k-1\}$, and $u = r(u_k)$. 

Clearly, if there are no node failures, then 
only one of the following cases may occur: either a leader $u$ exists 
with $u \in R(v)$, or $R(v)$ has a cycle. We prove the following lemma.
\begin{lemma}
\label{cycle}
For an arbitrary but fixed connector $v$, the set $R(v)$ has a
cycle with probability $O\left(\frac{2^{2\sqrt{\log n}} \log^2 n}{n}\right)$. 
Furthermore, the size of $R(v)$ is  $\abs{R(v)} = O(2^{\sqrt{\log n}} \log n)$, \whp, and 
$\abs{R(v)} = O(2^{\sqrt{\log n}})$, with constant probability.
\end{lemma}
\begin{proof}
Let $P(v)$ be a directed path $(v,u_1, \dots , u_k)$, 
where $u_1 = r(v)$, $u_{i+1} = r(u_i)$ for any 
$i \in \{ 1, \dots , k-1\}$, and $u_i \neq u_j,v$ for 
any $i,j \in \{ 1, \dots , k\}$, $i \neq j$. Then, 
$r(u_k) \in \{ v,u_1, \dots , u_{k-1} \}$ with 
probability $k/(n-1)$. Let this event be denoted by 
$A_k$. Furthermore, let $B_k$ be the event that 
$r(u_k)$ is not a leader ($B_0$ is the event that neither $r(v)$ is not a leader).
If $L$ is the set of leaders, then since communication partners are selected independently we have
\begin{eqnarray*}
&& Pr[\overline{A_k} \wedge B_k ~|~ \overline{A_1} \wedge B_1 
\dots \overline{A_{k-1}} \wedge B_{k-1}] = \frac{n-|L|-k}{n-1} \\
&&\mbox{ and } \\
&& Pr[\overline{A_1} \wedge B_1] = \frac{n-|L|}{n-1} \cdot \frac{n-|L|-1}{n-1}.
\end{eqnarray*}
Simple application of Chernoff bounds imply that $\abs{L} = n (1 \pm o(1))/2^{\sqrt{\log n}}$, \whp 
We condition on the event that this bound holds on 
$\abs{L}$, and obtain for some $k > c \cdot 2^{\sqrt{\log n}} \log n$ that 

\begin{eqnarray}
&& Pr[\overline{A_1} \wedge B_1] \cdot 
Pr[\overline{A_2} \wedge B_2 ~|~ \overline{A_1} \wedge B_1]
\cdot \dots \cdot \notag\\
&& \cdot 
Pr[\overline{A_k} \wedge B_k ~|~ \overline{A_1} \wedge B_1 
\wedge \dots \wedge 
\overline{A_{k-1}} \wedge B_{k-1}] \notag\\
&\leq& 
\left( 1- \frac{1}{2^{\sqrt{\log n}}}\right)^{c \cdot  
2^{\sqrt{\log n}} \log n} \leq n^{-3-\Omega(1)} ,
\label{eq-lengthR(v)}
\end{eqnarray}
whenever $c$ is large enough. The first inequality 
follows from $\abs{L} = \omega(k)$. This implies that 
the size of $R(v)$ is at most $c \cdot 
2^{\sqrt{\log n}} \log n$, \whp
Applying Inequality (\ref{eq-lengthR(v)}) with 
$k = c \cdot 2^{\sqrt{\log n}}$, we obtain that 
the size of  $R(v)$ is at most $c \cdot 
2^{\sqrt{\log n}}$, with some constant probability tending to 1 as 
$c$ tends to $\infty$. 

Now we prove that 
$$
Pr[R(v) \mbox{ contains a cycle}] =O\left(\frac{2^{2\sqrt{\log n}} \log^2 n}{n} \right) .
$$

We know that 
$$
Pr[A_i~|~\overline{A_0} \wedge B_0 
\wedge \dots \wedge 
\overline{A_{i-1}} \wedge B_{i-1} ]
= \frac{i}{n-1} ,
$$ 
where $B_0$ is the event that $r(v) \not\in L$ and $A_0 = \emptyset$.
Then, $\abs{R(v)}$ has a cycle, with probability less than
\begin{eqnarray*}
&& \sum_{i=1}^{n-\abs{L}-1} Pr[A_i~|~\overline{A_0} \wedge B_0 
\wedge \dots \wedge 
\overline{A_{i-1}} \wedge B_{i-1} ] \\
&& \cdot Pr[\overline{A_0} \wedge B_0 
\wedge \dots \wedge
\overline{A_{i-1}} \wedge B_{i-1} ]\\
&\leq& \sum_{i=1}^{c 2^{\sqrt{\log n}} \log n} 
Pr[A_i~|~\overline{A_0} \wedge B_0 
\wedge \dots \wedge
\overline{A_{i-1}} \wedge B_{i-1} ] + \\
&&
\sum_{i=c 2^{\sqrt{\log n}} \log n+1}^{n-\abs{L}-1} 
Pr[\overline{A_0} \wedge B_0 
\wedge \dots \wedge
\overline{A_{i-1}} \wedge B_{i-1} ] \\
&& \leq \frac{(c 2^{\sqrt{\log n}} \log n)^2}{n} + O(n^{-2 - \Omega(1)}). 
\end{eqnarray*}
As already shown, if $i > c 2^{\sqrt{\log n}} \log n$, then 
$Pr[\overline{A_1} \wedge B_1 
\wedge \dots \wedge 
\overline{A_{i-1}} \wedge B_{i-1} ] =
O(n^{-2 - \Omega(1)})$ if $c$ is large 
enough.
\qed
\end{proof}

From the previous lemma we obtain the following corollary.
\begin{corollary}
\label{cor_cycle}
Assume there are no node failures.
After phase 2, every connector stores the address of at least $2$ leaders, with 
probability at least $1-n^{-2}$.
\end{corollary}

We can also show the following upper bound on the number of connectors sharing the 
same leader address. This bound also holds in the case of node failures, since 
failed nodes can only decrease the number of connectors sharing the same leader address.
\begin{lemma}
\label{upper-bound}
Each connector shares the same leader address with $O(2^{3.1 \sqrt{\log n}} )$ other connectors, 
\whp
\end{lemma}
\begin{proof}
Let $S$ be a set of nodes, and let $r(S) = \{v \in V~|~r(v) \in S\}$. We model the parallel process 
of choosing nodes in the first round of a fixed sub-phase by the following sequential 
process (that is, the first round of the sub-phase is modeled by the whole sequence of steps 
of the sequential process). In the first step of the sequential process, 
all connectors choose a random node. We keep all edges 
between $(u,r(u))$ with $r(u) \in L$, and release all other edges. Let $L_1$ 
denote the set of nodes $u$ with $r(u) \in L$. In the $i$th step, we let each node of 
$V \setminus \cup_{j=0}^{i-1} L_j$ choose a node from the set $V \setminus \cup_{j=0}^{i-2} L_j$
uniformly at random, where $L_0 = L$. Clearly, the nodes are not allowed to choose themselves.
Then, $L_i$ is the set of nodes $u$ with $r(u) \in L_{i-1}$, and all edges 
$(u,r(u))$ (generated in this step) with $r(u) \not\in L_{i-1}$ are released.

Obviously, the sequential process produces the same edge distribution on the nodes 
of the graph as the parallel process. If now $S \subset L_{i-1}$, then the probability 
for a node $v \in V \setminus \cup_{j=0}^{i-1} L_j$ to choose a node in $S$ is 
$|S|/|V \setminus \cup_{j=0}^{i-2} L_j|$. Then, according to \cite{raab1998balls} the number 
of nodes $v$ with $r(v) \in S$ is at most 
$|S| + O(\log n + \sqrt{|S| \log n})$, \whp

Similar to the definition of $L_i$, 
for a leader $u$ the nodes $v$ with $r(v) =u$ are in set $L_1(u)$, the 
nodes $v$ with $r(r(v)) =u$ are in set $L_2(u)$, and generally, the nodes $v$ 
with $r(v) \in L_{i-1}(u)$ define the set $L_i(u)$.
Then, according to the 
arguments above
$|L_{i+1}(u)| = |L_{i}(u)| + O(\log n + \sqrt{|L_{i}(u)| \log n}) ,$
\whp We assume now that $|L_1(u)| = \Theta(\log n)$ (from \cite{raab1998balls} we may conclude that 
$|L_1(u)| =O(\log n)$, \whp). Then, for any $i \leq c \cdot 2^{\sqrt{\log n}} \log n$, we assume the
highest growth for $|L_{i+1}(u)|$, i.e.,
$|L_{i+1}(u)| = |L_{i}(u)| + O(\sqrt{|L_i(u)| \log n}) $, where $c$ is some constant. 
This recursion yields
$|L_{i+1}(u)|$ $\leq c (i+1)^2 \log n$, if $c$ is large enough.
Then,
$|L_{c \cdot 2^{\sqrt{\log n}} \log n}(u)|$ $< c^3 2^{2 \sqrt{\log n}} \log^3 n .$
Since $|R(v)| =O(2^{\sqrt{\log n}} \log n)$ for any $v$ (cf. Lemma \ref{cycle}), 
and assuming that $|L_{i}(u)| \leq c i^2 \log n$ for each $i$, 
we obtain the claim.
\qed
\end{proof}

Let us fix a sub-phase. We allow now node failures (i.e., each node may fail 
with some probability $O(1/(2^{\sqrt{\log n}}))$), and prove the following lemma.
\begin{lemma}
\label{no-followers}
There are $cn$ connectors, where $c>0$ is a constant, which 
store the addresses of at least two leaders, 
and each of these leader addresses is shared by at least 
$\Omega\left(\frac{2^{\sqrt{\log n}}}{\log n}\right)$ connectors, \whp
\end{lemma}
\begin{proof}
First, we consider the case in which no node failures are allowed. Then, we extend the proof.
Now let us assume that no failures occur.
We have shown in Lemma \ref{cycle} 
that the length of a path $(v, u_1, \dots , u_k, u)$ from a node $v$ to a leader $u$
is $O(2^{\sqrt{\log n}} \log n)$, \whp,
where $u_1 = r(v)$, $u_i = r(u_{i-1})$ for any $i \in \{ 2, \dots , k\}$, and 
$u = r(u_k)$. Let $u$ be a leader, and let $L_i(u)$ be the set of connectors which have distance 
$i$ from $u$ after a certain (arbitrary but fixed) sub-phase of the second phase. 
Furthermore, let $L_i(L) = \cup_{u \in L} L_i(u)$. For our analysis, 
we model the process of choosing nodes 
in the first step of this sub-phase by a sequential process (similar to the 
proof of the previous lemma), in which first
$v$ chooses a node, then $r(v)$ chooses a node, then $r(r(v))$ chooses a node, etc... In step
$i$ of this sequential process the $i$ node $u_{i-1}$ on the path $P(v)$ chooses a node.  
For some $i =O(2^{\sqrt{\log n}}/\log n)$ we have
$$
Pr[v \not\in \cup_{j=1}^{i} L_j(L)~|~\overline{A_1} \wedge \dots \wedge \overline{A_{i-1}}] 
\geq \left( 1-\frac{\abs{L}}{n-i-1}\right)^i ,
$$
Since $Pr[v \in \cup_{j=1}^{n-1} L_j(L)] = 1- O(2^{2 \sqrt{\log n}} \log^2 n/n)$  (cf.~Lemma 
\ref{cycle}), we obtain that,
given $R(v) \cap L \neq \emptyset$ (note that the number of nodes satisfying this 
property is $n(1-o(1))$, \whp),
a node has a path of length  $\Omega(2^{\sqrt{\log n}} / \log n)$
to a leader with probability $1-o(1)$, and thus the expected number of such nodes is n(1-o(1)).

Now we consider node failures. A node $v$ is considered failed, if it fails (as described 
at the beginning each node fails with probability $O(1/2^{\sqrt{\log n}})$), 
or there is a node in $R(v)$,
which fails. 
Since $|R(v)| = O(2^{\sqrt{\log n}})$ with constant probability, there is a node 
of such an $R(v)$ that fails with at most some constant probability. However, these probabilities 
are not independent. Nevertheless, the expected number of nodes, which will not be 
considered failed {\bf and} have a path of length $\Omega\left(\frac{2^{\sqrt{\log n}}}{\log n}\right)$
to a leader, is $\Theta(n)$.

Now, consider the 
following Martingale sequence. Let $v_1, \dots , v_{n-\abs{L}}$ denote the connectors. 
In step $j$, we reveal the directed edges and nodes from 
node $v_j$ to all nodes in all $R(v_j)$ obtained from the different sub-phases. 
Given that $\abs{R(v_j)} = O(2^{\sqrt{\log n}} \log n)$, 
we apply the Azuma-Hoeffding inequality \cite{mitzenmacher05probability}, and 
obtain that $\Theta(n)$ nodes are connected to a leader by a path of 
length $\Omega\left(\frac{2^{\sqrt{\log n}}}{\log n}\right)$ and will not be considered failed, \whp

Summarizing, a $\Theta(n)$ fraction of the nodes 
store at the end of the first 
phase the addresses of at least two leaders, and such a connector shares each of these addresses 
with $\Omega(2^{\sqrt{\log n}} / \log n)$ other connectors, \whp
\qed
\end{proof}

Applying pointer jumping on all connectors as described in the algorithm, we obtain the 
following result.
\begin{observation}
\label{obs_pointer_jumping}
If in an arbitrary but fixed sub-phase of the second phase $R(v) \cap L \neq \emptyset$ for 
some connector $v$, then $v$ stores the address of a leader $u$ at the end of this phase,
\whp
\end{observation}
This observation is a simple application of the pointer jumping algorithm~ \cite{leighton1992introduction}
on a directed path of length $|R(v)|$. 
According to Lemma \ref{cycle}, $|R(v)|=O(2^{\sqrt{\log n}} \log n)$, \whp

Now we concentrate on the third phase. We condition on the event that each connector has stored
at least two and at most $5$ different leader addresses. Furthermore, an address stored by a 
connector is shared with at least $\Omega(2^{\sqrt{\log n}} / \log n)$ other connectors, with
high 
probability (see Lemma \ref{no-followers}). Out of these connectors, let $C$ be the set of nodes $v$ with the following 
property. The first time a leader of $v$ receives the message, $v$ will contact this leader 
in the next step, pulls the message, and in the next step 
it will push the message to the other leader. 
Clearly, for a node $v$ this 
event occurs with constant probability, independently of the other nodes. Therefore, the 
total number of nodes in $C$ with at least two different leader addresses, where each of these 
addresses is shared by at least $\Omega(2^{\sqrt{\log n}} / \log n)$ other connectors, is 
$\Theta(n)$, \whp We call the set of these nodes $\tilde{C}$. Now we have the following observation.

\begin{observation}
\label{obs_distr_good}
Let $C_i$ be the set of nodes which store the same (arbitrary but fixed)
leader address after a certain subsphase, and 
assume that $\abs{C_i} = \Omega(2^{\sqrt{\log n}} / \log n)$. Then, $\abs{C_i \cap \tilde{C}} 
= \Theta(\abs{C_i})$,
\whp
\end{observation} 
The proof of this observation follows from the fact that if two nodes share the same 
address after a certain subphase, then each of these nodes will share with 
probability $1-o(1)$ a leader address obtained 
in some other subphase with at least $\Omega(2^{\sqrt{\log n}} / \log n)$ other connectors.
However, these events are not independent. Let now $C_j$ be some other set, which 
contains a node $v \in C_i$. Since $|C_i|, |C_j| = O(2^{3.1 \sqrt{\log n}})$ 
(see Lemma \ref{upper-bound}), there will be with probability at least $1-n^{-2}$ at most
$4$ nodes in $C_i \cap C_j$. Conditioning on this, we apply 
for the nodes of $C_i \cap C$ the same Martingale sequence 
as in the proof of Lemma \ref{no-followers}. By taking into account that 
in this case the Martingale sequence 
satisfies the $4$-Lipschitz condition (the nodes of $C_i$ are 
part of the Martingale only), we obtain the statement of the observation.

Now we are ready to show the following lemma. 
\begin{lemma}
\label{lem-growth}
After the third phase the number of informed nodes is at least $\frac{n}{2^{7\sqrt{\log n}}}$, \whp
\end{lemma}
\begin{proof}
For a node $v \in \tilde{C}$, let $C_v^{(1)}$ and $C_v^{(2)}$ represent two sets of nodes, which
store the same leader address as $v$ (obtained in the same 
sub-phases of the second phase), and for which we have 
$\abs{C_v^{(1)}}, \abs{C_v^{(2)}} = \Omega(2^{\sqrt{\log n}} / \log n)$.  
We know that each node has exactly $2$ leader addresses. Since after phase $0$ 
at least $\log^2 n$ nodes are informed, we may assume that 
at the beginning of this phase a node $w \in \tilde{C}$ is informed, and $w$ pushes the message 
exactly once. That is, after two steps 
all nodes of $C_w^{j} \cap \tilde{C}$ are informed, where $j$ is either $1$ or $2$
(we may assume w.l.o.g.~that $j=1$). Furthermore, 
we assume that these are the only nodes which are informed after the second step.

Now, we show by 
induction that the following holds. After $2i$ steps, the number of informed nodes $I(i)$ in 
$\tilde{C}$ is at least $\min\{2^{\sqrt{\log n} \cdot i/2}, n/2^{7\sqrt{\log n}}\}$, \whp 
Furthermore, 
there is a partition of the set $\{C_v^{(j)} \cap \tilde{C}~|~v \in I(i),~j \in \{ 1,2\}\}$, 
into the sets $E^{(j)}(i)$ and $F^{(j)}(i)$, where 
$E^{(j)}(i)$ are the sets $C_v^{(j)} \cap \tilde{C}$ with $\abs{C_v^{(j)} \cap \tilde{C} \cap I(i)} = 
O(\log n)$, and $F^{(j)}(i)$ are the sets $C_v^{(j)} \cap \tilde{C}$ with $C_v^{(j)} 
\cap \tilde{C} \cap I(i) = C_v^{(j)} \cap \tilde{C}$. Roughly speaking, the sets 
belonging to $E^{(j)}(i)$ contain some nodes, which have just been informed in the last time step, 
and most of the nodes from these sets are still uninformed. If now these nodes perform \push, 
and in the next step 
the nodes of the sets in $E^{(j)}(i)$ a \pull, then these nodes become informed as well.
Our assumption is that the number of sets $E_v^{(j)}(i)$ 
is $\Omega(\abs{I(i)}/\log n)$, \whp
This obviously holds before the first or
after the second step. 

Assume that the induction hypothesis holds after step $2i$ and we are 
going to show that it also holds after step $2(i+1)$. 
Clearly, if $U$ is some set of nodes 
which have the same leader address after an arbitrary but fixed subphase of the second phase, where 
$|U| = \Omega(2^{\sqrt{\log n}} / \log n)$, then we have $|U \cap \tilde{C}| = \Theta(|U|)$, \whp
(see Observation \ref{obs_distr_good}). 
On the other hand, there are at least $\Omega(n/2^{3.1 \sqrt{\log n}})$
such sets $U$ with $U \not\in \cup_{j=1,2} F^{(j)}(i)$, \whp, since the 
largest set we can 
obtain has size $O(2^{3.1\sqrt{\log n}})$, \whp (cf. Lemma \ref{upper-bound}). According to our 
induction hypothesis, 
at least $\Omega(|I(i)|/\log n)$ and at most $O(|I(i)|)$ 
of these sets are elements of $E^{(j)}(i)$, where $v \in I(i)$. 

Clearly, a node $v \in \tilde{C} \setminus I(i)$ will be in at most one of these sets, \whp
Since any of these sets accomodates at least $\Theta(2^{\sqrt{\log n}}/\log n)$ nodes from 
$\tilde{C}$, \whp,
the number of informed nodes increases within two steps by at least a factor of 
$\Theta(2^{\sqrt{\log n}}/\log^2 n) \gg 2^{\sqrt{\log n} /2} ,$
which leads to $|I(i+1)| \geq 2^{\sqrt{\log n} \cdot (i+1)/2}$, \whp The induction step can be 
performed as long as $|I(i)| \leq n/2^{7\sqrt{\log n}}$. 
Now we concentrate on the distribution of these nodes among the sets $U \not\in \{ 
E_v^{(j)}(i)~|~v \in I(i),~j \in \{ 1,2\}\}$. Note that each such node belongs to two sets; 
one of these sets is an element of $E_v^{(j)}(i)$ for some $v \in I(i)$, while the other one is not. 
Since the total number of nodes in some set of $E^{(j)}(i)$ is $O(2^{3.1 \sqrt{\log n}})$, \whp, we have 
$|I(i+1)| = O(2^{3.1 \sqrt{\log n}} \cdot |I(i)|) = O(n/2^{3.9 \sqrt{\log n}}) .$
As argued above, there are at least $\Omega(n/2^{3.1 \sqrt{\log n}})$
sets $U$ with $U \not\in \{ F_v^{(j)}(i+1)~|~v \in I(i+1),~j \in \{ 1,2\}\}$, \whp, where 
$U$ is some set of nodes 
which have the same leader address after an arbitrary but fixed subphase of the second phase, and 
$|U| = \Omega(2^{\sqrt{\log n}} / \log n)$. Thus,  a node $v \in (I(i+1) \setminus I(i)) \cap \tilde{C}$
is assigned to a fixed such $U$ with probability $O(1/|I(i+1)|)$. 
Therefore, none of the sets $E_v^{(j)}(i+1)$ will accomodate more than $O(\log n)$ nodes from 
$(I(i+1) \setminus I(i)) \cap \tilde{C}$, \whp\cite{raab1998balls}, and the claim follows.
\qed
\end{proof}

Now we are ready to prove our main theorem, which also compares the communication overhead of the 
usual \ex algorithm of \cite{karp2000randomized} to our algorithm. Note that the {\bf bit} 
communication complexity of \cite{karp2000randomized} w.r.t.~one rumor 
is $O(n b \cdot \log \log n)$, \whp, where $b$ is the bit length of that rumor. We should also 
mention here that in \cite{karp2000randomized} the authors assumed that messages (so called 
updates in replicated data-bases) are frequently generated, and thus the cost of opening 
communication channels amortizes over the cost of sending messages through these channels.
If in our scenario messages are frequently generated, 
then we may also assume that the cost of the pointer jumping phase is negligable 
compared to the cost of sending messages, and thus the communication overhead in our case 
would also be $O(n b \log \log n)$. In our theorem, however, we assume that one message has 
to be distributed, and sending the IP-address of a node through a communication channel 
is $O(\log n)$. Also, opening a channel without sending messages generates an
$O(\log n)$ communication cost.

\begin{theorem}
\label{main_theo}
At the end of the JPP algorithm, all but O(F) nodes are informed \whp, where 
$F$ is the number of failed nodes as described above. 
The algorithm has running time $O(\sqrt{\log n})$
and produces a {\bf bit} communication complexity of 
$O(n (\log^{3/2} n + b \cdot \log \log n))$, \whp, where $b$ is the 
bit length of the message.  
\end{theorem}
\begin{proof}
In the fourth phase we apply the (median counter) algorithm presented 
in \cite{karp2000randomized}. 
For the sake of completeness, we describe this algorithm here as given in 
\cite{karp2000randomized}. There, each node can be in a state called $A$, $B$, $C$, or $D$. State 
$B$ is further 
subdivided in substates $B_1$, \dots, $B_{ctr_{\max}}$, where $ctr_{\max} = O(\log \log n)$ 
is some suitable integer. At the beginning of this phase, all informed nodes are in state $B_1$ and 
all uninformed nodes are in state $A$. The rules are as follows:
\begin{itemize}
\item If a node $v$ in state $A$ receives the rumor only from nodes in state $B$, then 
it switches to state 
$B_1$. If $v$ obtains the rumor from a state $C$ node, then it switches to state $C$.
\item If a node $v$ in state $B_i$ communicates with more nodes in some state $B_j$ with 
$j \geq i$ than with nodes in state $A$ or $B_{j'}$ with $j' < i$, 
then
$v$ switches to state $B_{i+1}$. If $v$ gets the rumor from a state $C$ node, then it switches to 
state $C$.
\item A node in state $C$ sends the rumor for $O(\log \log n)$ further steps. Then, it switches to state 
$D$ and stops sending the rumor. 
\end{itemize}
We know that at the end of the third phase, there are at least $n/2^{7 \sqrt{\log n}}$ informed nodes, 
\whp (cf.~Lemma \ref{lem-growth}). In order to apply Theorem 3.1 of \cite{karp2000randomized}, we have 
to couple 
the original median counter algorithm with our algorithm. Let $I(t_0)$ be the set of informed nodes 
at the end of the third phase. Clearly, the communication overhead w.r.t.~the rumor is $O(n \cdot b)$ 
in the third phase, since 
each connector transmits at most twice the message, and the number of leaders is 
bounded by $O(n/2^{\sqrt{\log n}})$, \whp Then, there is a time step in the original median counter
algorithm such that the number of informed nodes is $|I(t_0)|$ too\footnote{The time step, in which 
more than $|I(t_0)|$ are informed for the first time, is subdivided, such that we have a time 
step, in which there are exactly $|I(t_0)|$ nodes informed.}.
Obviously, there might exist 
nodes at this time step, which are in some state $B_j$, with $j > 1$, $C$, or $D$. 

At this time step, we couple the random choices of the nodes in the two algorithms. As long as
$|I(i)| \leq n/\log^2 n$, it holds that $|I(i+1)| > (1+\epsilon) |I(i)| $, \whp (see exponential 
growth phase in Theorem 3.1, \cite{karp2000randomized}), for some constant 
$\epsilon > 0$, and the number of informed nodes (as well as the constant $\epsilon$)
produced by our algorithm dominates the 
number of infomed nodes in the original median counter algorithm. This holds since at time step 
$t_0$ we only have state $B_1$ or $A$ nodes in our algorithm, 
while the original median counter algorithm may 
contain state $B_j$ and $C$ nodes at that time step, where $j>1$. Therefore, these nodes will stop
earlier sending the message. When $|I(i)| \geq n/\log^2 n$ for the first time, the communication 
overhead w.r.t.~the rumor is bounded by $O(n \cdot b)$. 

Once the message is distributed to $n/\log^2 n$ nodes, one needs $O(\log \log n)$ additional steps 
to disseminate the rumor among all vertices of the graph (see quadratic shrinking phase in 
Theorem 3.1, \cite{karp2000randomized}). Moreover, all nodes stop sending the rumor after 
$O(\log \log n)$ 
additional steps, once all nodes are informed (cf.~Theorem 3.1, \cite{karp2000randomized}). 
Thus, the total communication overhead w.r.t.~the rumor is bounded 
by $O(n b \cdot \log \log n)$, \whp

The communication overhead w.r.t.~the addresses sent by the nodes in the pointer jumping phase is 
upper bounded by $O(n \sqrt{\log n} \cdot 
\log n)$, where $\sqrt{\log n}$ stands for the number of steps in 
the second phase, while the $\log n$ term describes the bit size of a message  
(an address is some polynomial in $n$).
\qed
\end{proof}

\section{Discussion - Non-exact Case}
\label{disc}

As mentioned in Section \ref{algo}, a modified version of our algorithm also 
works if the nodes only have an estimate of $\log n$, which is accurate up to some 
constant factor. In this case, we introduce some dummy sub-phases between any two phases and 
any sub-phases of phase $2$. Now, for a node $v$ the length of sub-phase $i$ of phase $2$ will be 
$\rho^{2i}c \sqrt{\log n_v}$, and between sub-phase $i$ and $i+1$, there will be 
a dummy sub-phase of length $\rho^{2i+1} c \sqrt{\log n_v}$. Here $n_v$ is the estimate 
of $n$ at node $v$.
Accordingly, the dummy sub-phase between 
phase $1$ and $2$ will have length $\rho c \sqrt{\log n_v}$, between phases $2$ and $3$ 
length $\rho^{11} c \sqrt{\log n_v}$, and between $3$ and $4$ length $\rho^{13} c \sqrt{\log n_v}$.
The length of phase $3$ will be $\rho^{12} c \sqrt{\log n_v}$, and that of 
phase $4$ will be $\rho^{14} c \sqrt{\log n_v}$. Here $\rho$ will be a large constant, such that 
$\rho^i \gg \sum_{j=0}^{i-1} \rho^{j}$ for any $i < 15$. Furthermore, 
\begin{eqnarray*}
\sum_{j=0}^{i} \rho^j c \min_{v \in V} \sqrt{\log n_v} &\gg&
\sum_{j=0}^{i-1} \rho^j c \max_{v\in V} \sqrt{\log n_v} + \\
&& c \max_{v\in V} \sqrt{\log n_v} ,
\end{eqnarray*}
where $i \in \{ 1, \dots , 15\}$. 

The role of the dummy sub-phases is to synchronize 
the actions of the nodes. That is, no node will enter a phase or sub-phase before 
the last node leaves the previous phase or sub-phase. Accordingly, no node will leave a phase 
or a sub-phase, before the last node enters this phase or sub-phase. Moreover, the whole 
set of nodes will be together for at least $c \sqrt{\log n}$ steps in every phase or sub-phase.
This ensures that all the phases and sub-phases of the algorithm will work correctly, and 
lead to the results we have derived in the previous section. Note that, however, the 
communication overhead might increase to some value $O(n(\log^{3/2} n + b \sqrt{n})$. 




\bibliographystyle{acm}
\bibliography{rumor}

\appendix


\end{document}